\documentclass[a4paper]{article}
\usepackage[american]{babel}
\usepackage{amsmath,amsthm,amsfonts}
\usepackage{graphicx}
\usepackage{paralist}

\newtheorem{theorem}{Theorem}[section]

\newcommand{\f}{\mathbf{f}}
\newcommand{\g}{\mathbf{g}}
\newcommand{\norm}[2]{\|#2\|_{#1}}
\newcommand{\qmax}{q_\textup{max}}
\newcommand{\R}{\mathbb{R}}
\newcommand{\rhomax}{\rho_\textup{max}}

\newcommand{\unit}[1]{\textup{#1}}
\newcommand{\vmax}{V_\textup{max}}

\graphicspath{{./}{pictures/}}

\title{Fundamental diagrams for kinetic equations of traffic flow}
\author{	Luisa Fermo \\
		{\small\it Department of Mathematics and Computer Science}\\[-1mm]
		{\small\it University of Cagliari}\\[-1mm]
		{\small\it Viale Merello 92, 09123 Cagliari, Italy}\\[5mm]
		Andrea Tosin\thanks{Partially supported by the European Commission under the 7th Framework Program (grant No. 257462 HYCON2 Network of Excellence)
			and by the Google Research Award ``Multipopulation Models for Vehicular Traffic and Pedestrians'', 2012--2013.} \\
		{\small\it Istituto per le Applicazioni del Calcolo ``M. Picone''}\\[-1mm]
		{\small\it Consiglio Nazionale delle Ricerche}\\[-1mm]
		{\small\it Via dei Taurini 19, 00185 Roma, Italy}
		}
\date{}

\begin{document}
\maketitle

\begin{abstract}
In this paper we investigate the ability of some recently introduced discrete kinetic models of vehicular traffic to catch, in their large time behavior, typical features of theoretical fundamental diagrams. Specifically, we address the so-called ``spatially homogeneous problem'' and, in the representative case of an exploratory model, we study the qualitative properties of its solutions for a generic number of discrete microstates. This includes, in particular, asymptotic trends and equilibria, whence fundamental diagrams originate.

\medskip


\medskip

\end{abstract}

\section{Quick overview of discrete kinetic models}
\label{sect:intro}
In the research about vehicular traffic, the \emph{fundamental diagram} is a relationship linking the flux $q$ of vehicles to their density $\rho$ in steady flow conditions. In particular, the density, expressing the number of vehicles per kilometer of road, is implicitly supposed to be uniformly distributed in space, so that a single value can be associated to the whole road. Another relationship is the \emph{speed diagram}, which links the mean speed $u$ of cars to their density and is derived from the fundamental diagram according to the equation:
\begin{equation*}
	u=\frac{q}{\rho}. 
\end{equation*}

Fundamental and speed diagrams, collectively referred to also as \emph{fundamental relationships}, provide a synthetic macroscopic description of traffic trends at equilibrium. Generally speaking, the mapping $\rho\mapsto u=u(\rho)$ is non-increasing from a maximum value, attained for $\rho\to 0^+$, to $u=0$, attained for $\rho=\rhomax>0$, the latter being the maximum density of cars that can be accommodated on the road. Consequently, the flux $q$ vanishes for both $\rho\to 0^+$ and $\rho=\rhomax$. Moreover, the mapping $\rho\mapsto q=q(\rho)$ normally features a single global maximum $\qmax>0$, called \emph{road capacity}, which is attained for a critical density value $\sigma$ strictly comprised between $0$ and $\rhomax$. The flow regime for $0\leq\rho\leq\sigma$ is called \emph{free}, whereas that for $\sigma<\rho\leq\rhomax$ is called \emph{congested}.

Fundamental relationships are either measured experimentally, cf. e.g.,~\cite{bonzani2003MCM,kerner2004BOOK}, or expressed analytically, cf. e.g.,~\cite{garavello2006BOOK,piccoli2009ENCYCLOPEDIA}. See also~\cite{li2011TRR}, where a systematic study is developed aiming at an analytical characterization of fundamental diagrams out of experimental data. Analytical relationships are mostly used in the mathematical theory of vehicular traffic for closing \emph{first order macroscopic models}, i.e., fluid dynamic models based on the continuity equation:
\begin{equation}
	\frac{\partial\rho}{\partial t}+\frac{\partial q}{\partial x}=0,
	\label{eq:continuity}
\end{equation}
$x$, $t$ being the independent variables denoting space and time, respectively, which expresses the conservation of cars seen as a flowing continuum. Indeed an analytical fundamental diagram $q=q(\rho)$, or alternatively a speed diagram $u=u(\rho)$, enables one to transform Eq.~\eqref{eq:continuity} in a self-consistent equation for the density $\rho$:
\begin{equation}
	\frac{\partial\rho}{\partial t}+\frac{\partial}{\partial x}q(\rho)=0
	\qquad \textup{or, equivalently,} \qquad
	\frac{\partial\rho}{\partial t}+\frac{\partial}{\partial x}(\rho u(\rho))=0.
	\label{eq:first.order}
\end{equation}
Such a closure contains implicitly a conceptual approximation. In fact, as said before, fundamental relationships are meaningful in principle only in uniform steady flow conditions. Plugging them into Eq.~\eqref{eq:continuity}, which is expected to describe the evolution of the car density also far from equilibrium, requires the further (partly arbitrary) assumption that equilibrium conditions always hold at least locally in time and space, so that either relationship $q(t,\,x)=q(\rho(t,\,x))$, $u(t,\,x)=u(\rho(t,\,x))$ be admissible for every $x$ and $t$.

Besides the aspects just set forth, there is another questionable point concerning the use of fundamental relationships as constitutive laws in traffic flow models. In real traffic, fundamental relationships are a consequence, \emph{not} a cause, of vehicle dynamics. In other words, cars do not move because of the ``law of fundamental diagram'' as e.g., falling objects instead do because of gravitational laws. Vehicle dynamics are rather triggered by one-to-one, or at most one-to-few, interactions among cars, which take place at a lower microscopic scale and then generate collective trends visible at a larger macroscopic scale. In addition, microscopic dynamics are not even fully driven by classical mechanical rules, because the presence of drivers induces subjective decisions which should be modeled by duly taking into account their intrinsic level of randomness.

A possible way of addressing these issues is to move to a scale of representation of traffic flow different from the macroscopic one. However, aiming ultimately still at the description of collective trends, which are also more amenable to mathematical analysis than individual behaviors of cars, we refrain from arriving at the detail of the microscopic scale and settle instead at the \emph{mesoscopic} (or \emph{kinetic}) level. Particularly, we refer to some recently introduced models~\cite{bellouquid2012M3AS,delitala2007M3AS,fermo2013SIAP} in which the space of microstates of cars, namely the spatial position $x$ and the speed $v$, is partly or fully discrete in order to incorporate in the kinetic representation the intrinsic microscopic \emph{granularity} of the vehicle distribution.

Basically, the discrete kinetic representation of vehicular traffic consists in selecting a certain number, say $n$, of microscopic \emph{speed classes}:
\begin{equation*}
	0=v_1,\, v_2,\, v_3,\,\dots,\,v_n=\vmax>0, \qquad v_j<v_{j+1}\quad\forall\,j=1,\,\dots,\,n-1
\end{equation*}
and then in identifying vehicles traveling in the $j$-th speed class by means of their \emph{statistical distribution function} $f_j=f_j(t,\,x)$. More precisely, by definition $f_j$ is such that $f_j(t,\,x)dx$ is the infinitesimal number of vehicles that at time $t$ are comprised between $x$ and $x+dx$ and travel at speed $v_j$. The \emph{density} $\rho=\rho(t,\,x)$ of cars in the point $x$ of the road at time $t$ is obtained by summing the distributions functions over all speed classes:
\begin{equation}
	\rho(t,\,x)=\sum_{j=1}^{n}f_j(t,\,x).
	\label{eq:rho}
\end{equation}
This quantity compares directly with the macroscopic density used in fluid dynamic models~\eqref{eq:first.order}. In addition, the macroscopic flux $q$ and the mean speed $u$ can be computed according to their original definition, i.e., as statistics over the microscopic speeds:
\begin{equation}
	q(t,\,x)=\sum_{j=1}^{n}v_jf_j(t,\,x), \qquad
	u(t,\,x)=\frac{q(t,\,x)}{\rho(t,\,x)}=\frac{1}{\rho(t,\,x)}\sum_{j=1}^{n}v_jf_j(t,\,x).
	\label{eq:q.u}
\end{equation}
The same definitions hold obviously also in uniform steady conditions, when the $f_j$'s are constant in both $x$ and $t$. Hence fundamental relationships $q=q(\rho)$, $u=u(\rho)$ can be genuinely obtained statistically at equilibrium rather than being empirically postulated \emph{a priori}.

The distribution functions $f_j$ are found as the solution to a system of $n$ partial differential equations obtained by balancing, in the space of microstates, the numbers of cars which gain and lose, respectively, a given test state $(x,\,v_j)$ in the unit time because of reciprocal interactions:
\begin{equation}
	\frac{\partial f_j}{\partial t}+v_j\frac{\partial f_j}{\partial x}=G_j[\f,\,\f]-f_jL_j[\f],
	\label{eq:kinetic.continuous.space}
\end{equation}
where $\f=(f_1,\,\dots,\,f_n)$. The left-hand side is a classical transport operator, whereas the right-hand side is the difference between the $j$-th \emph{gain} $G_j$ and \emph{loss} $L_j$ \emph{operators}. The latter are bilinear and linear operators, respectively, on $\f$, which encode \emph{stochastic} microscopic interaction dynamics among cars. For instance, in the models proposed in~\cite{bellouquid2012M3AS,delitala2007M3AS} they take the following forms:
\begin{align}
	\begin{aligned}[c]
		& G_j[\f,\,\f](t,\,x)=\sum_{h,\,k=1}^{m}\int_{x}^{x+\xi}\eta_{hk}(t,\,y)A_{hk}^j(t,\,y)f_h(t,\,x)f_k(t,\,y)w(y-x)\,dy \\
		& L_j[\f](t,\,x)=\sum_{k=1}^{m}\int_{x}^{x+\xi}\eta_{jk}(t,\,y)f_k(t,\,y)w(y-x)\,dy,
	\end{aligned}
	\label{eq:G-L.continuous.space}
\end{align}
where: $\eta_{hk}$ is the frequency of interaction between any two vehicles traveling at speeds $v_h$, $v_k$ (\emph{interaction rate}); $0\leq A_{hk}^j\leq 1$ is the probability that a vehicle traveling at speed $v_h$ (\emph{candidate} vehicle) changes its speed to $v_j$ (\emph{test} speed) when interacting with a vehicle traveling at speed $v_k$ (\emph{field} vehicle); $\xi>0$ is a length over which \emph{nonlocal} interactions between candidate and field vehicles are effective (\emph{interaction length}); and finally $w$ is a function with unit integral on $[0,\,\xi]$ weighting the interactions on the basis of the distance between the interacting vehicles.

The collection of the transition probabilities $\{A_{hk}^{j}\}_{h,\,k,\,j=1}^{m}$ is called the \emph{table of games}. The coefficients $A_{hk}^{j}$ are required to satisfy pointwise in time and space the following probability distribution property:
\begin{equation}
	\sum_{j=1}^{n}A_{hk}^{j}=1, \qquad \forall\,h,\,k=1,\,\dots,\,n,
	\label{eq:table.games.sum.1}
\end{equation}
which actually implies the conservation of cars in model~\eqref{eq:kinetic.continuous.space}. In fact, summing over $j$ both sides of Eq.~\eqref{eq:kinetic.continuous.space} and recalling the definitions~\eqref{eq:rho}-\eqref{eq:q.u} one obtains then the continuity equation~\eqref{eq:continuity}.

It is worth pointing out that the description of microscopic car interactions implemented in Eq.~\eqref{eq:G-L.continuous.space} is inspired by the \emph{stochastic game theory} for including the aforesaid randomness of driver behaviors. More precisely, vehicle interactions are regarded as games that two players (viz. drivers) play using their respective pre-interaction speeds $v_h$, $v_k$ as game strategies. The payoff of the game is the new speed class $v_j$ that the candidate driver shifts to after the interaction with the field driver. However, these games are stochastic since for each pair of game strategies the payoff is known only in probability due to the intrinsic uncertainty in forecasting human behaviors. Specifically, such a probability is the transition probability $A_{hk}^{j}=\mathbb{P}(v_h\to v_j\vert v_k)$, which can also depend on time and space, as shown in Eq.~\eqref{eq:G-L.continuous.space}, because in many models it is parameterized by the vehicle density $\rho$. This allows models to account for possible changes in the transition probabilities due to the evolution of the collective state of the system. In summary, modeling car interactions as stochastic games means describing, in probability, how candidate vehicles can get the test state and how the test vehicle can lose it because of the presence of field vehicles.

Still at a mesoscopic level, the model proposed in~\cite{fermo2013SIAP} pushes forward the discretization of the space of microstates by discretizing also the space variable. This is done in order to consider also the granularity of the space distribution of cars along a road. Hence the spatial domain is partitioned in pairwise disjoint cells $I_i$ and the statistical description of the vehicle distribution relies now on a collection of double-indexed distribution functions $f_{ij}=f_{ij}(t)$, each representing the number of cars which at time $t$ are in the $i$-th spatial cell and travel at speed $v_j$. Summing over $j$ for fixed $i$ gives the density $\rho_i=\rho_i(t)$ of cars in the $i$-th cell at time $t$. Notice that the $f_{ij}$'s, thus also $\rho_i$, are constant in $I_i$. In fact the modeled spatial granularity of traffic does not allow for a finer statistical detection of the positions of cars within a given cell at an ensemble level.

The $f_{ij}$'s are found as the solution to the following system of ordinary differential equations in time:
\begin{equation}
	\frac{df_{ij}}{dt}+\frac{v_j}{\ell}(\Phi_{i,i+1}f_{ij}-\Phi_{i-1,i}f_{i-1,j})=G_{ij}[\f,\,\f]-f_{ij}L_{ij}[\f],
	\label{eq:kinetic.discrete.space}
\end{equation}
where $\f=\{f_{ij}\}_{i,\,j}$, $\ell>0$ is the characteristic length of every space cell, and $\Phi_{i,i+1}$ is the \emph{flux limiter} at the interface between two adjacent cells, which regulates the percentage of cars that can actually flow from the $i$-th to the $(i+1)$-th cell on the basis of the number of incoming cars and the free room available to them in the destination cell. Notice that this term has no counterpart in the continuous-in-space model~\eqref{eq:kinetic.continuous.space}, as it is a genuine consequence of the inclusion of space granularity into the mesoscopic mathematical description. The structure of the gain and loss operators at the right-hand side is as follows:
\begin{align}
	\begin{aligned}[c]
		& G_{ij}[\f,\,\f](t)=\frac{\ell}{2}\sum_{h,\,k=1}^{n}\eta_{hk}(t,\,i)A_{hk}^{j}(t,\,i)f_{ih}(t)f_{ik}(t) \\
		& L_{ij}[\f](t)=\frac{\ell}{2}\sum_{k=1}^{n}\eta_{jk}(t,\,i)f_{ik}(t).
	\end{aligned}
	\label{eq:G-L.discrete.space}
\end{align}
They are again inspired by the stochastic game theory, the main differences with respect to Eq.~\eqref{eq:G-L.continuous.space} being that the space dependence is here discrete and that interactions among candidate and field vehicles are localized within a single space cell. Nevertheless, since the latter is not reduced to a point, interactions are ultimately still nonlocal in space, the characteristic cell length $\ell$ playing morally the role of the interaction length $\xi$. The coefficient $\frac{\ell}{2}$ is due to normalization purposes.

Starting from the kinetic traffic flow models just outlined, in the next two sections we investigate how and to what extent information about fundamental relationships can be extracted from them. In particular, in Section~\ref{sect:space.homog} we discuss the proper formalization of the so-called \emph{spatially homogeneous problem} leading to fundamental relationships, showing that, in the abstract, it is actually common to all models presented so far. Furthermore, we specialize it for an exploratory choice of the table of games, among the simplest conceivable ones. Next, in Section~\ref{sect:theory} we perform an asymptotic analysis of the equilibria of the spatially homogeneous problem, specifically aimed at proving the existence and uniqueness of the flux and speed diagrams. In particular, uniqueness for a generic number of speed classes is proposed here for the first time. As a by-product, we also provide explicit (recursive) analytical formulas of such diagrams.

\section{The spatially homogeneous problem}
\label{sect:space.homog}
As anticipated at the beginning of Section~\ref{sect:intro}, when speaking of fundamental diagrams the tacit assumption is that traffic flow is stationary and homogeneous in space. From the mathematical point of view, we can mimic these conditions in two subsequent steps.
\begin{itemize}
\item First, we account for spatial homogeneity by assuming that the kinetic distribution functions are independent of the variable representing the space, either $x$ or the index $i$. Hence we are left with $f_j=f_j(t)$, which stands for the number of cars that at time $t$ travel at speed $v_j$. Under this assumption, both models~\eqref{eq:kinetic.continuous.space}-\eqref{eq:G-L.continuous.space} and~\eqref{eq:kinetic.discrete.space}-\eqref{eq:G-L.discrete.space} reduce to:
\begin{equation}
	\frac{df_j}{dt}=\sum_{h,\,k=1}^{n}\eta_{hk}A_{hk}^{j}f_h f_k-f_j\sum_{k=1}^{n}\eta_{jk}f_k,
	\label{eq:spat.homog}
\end{equation}
up to incorporating into the interaction rate the coefficient $\frac{\ell}{2}$ appearing in Eq.~\eqref{eq:G-L.discrete.space}. In particular, to see why the advection term at the left-hand side of Eq.~\eqref{eq:kinetic.discrete.space} vanishes we notice that also the flux limiters become independent of $i$. In fact, spatially homogeneous conditions imply that the number of vehicles flowing from the incoming cell and the amount of free room in the destination cell are the same at each cell interface.
\item Second, we identify the stationary configurations of traffic flow with the \emph{stable} equilibria (if any) of the dynamical system~\eqref{eq:spat.homog}.
\end{itemize}

Recalling property~\eqref{eq:table.games.sum.1} of the table of games, it is immediate to check formally that Eq.~\eqref{eq:spat.homog} is such that
\begin{equation*}
	\frac{d}{dt}\sum_{j=1}^{n}f_j=0.
\end{equation*}
Therefore the (spatially homogeneous) density of cars $\rho=\sum_{j=1}^{n}f_j$ is conserved in time, whereby any possible equilibrium point $\f^\infty=\{f^\infty_j\}_{j=1}^{n}$ of system~\eqref{eq:spat.homog} is characterized by the fact that
\begin{equation*}
	\sum_{j=1}^{n}f_j^\infty=\sum_{j=1}^{n}f_j(0).
\end{equation*}
Consequently, it is possible to understand $\rho$ as a parameter of Eq.~\eqref{eq:spat.homog} (fixed by the initial condition) and to select the admissible equilibrium fluxes and mean speeds corresponding to a given density by looking at the large time behavior of the solutions to Eq.~\eqref{eq:spat.homog}. This is particularly useful when flux-density and speed-density relationships at equilibrium have to be computed numerically, for good numerical schemes tend naturally to converge, for large times, to (approximations of) stable equilibria of system \eqref{eq:spat.homog}. Conversely, finding automatically admissible equilibria by solving the nonlinear algebraic system resulting from equating to zero the right-hand side of Eq.~\eqref{eq:spat.homog} can be harder, because additional criteria are needed to distinguish stable from unstable solutions.

The mappings
\begin{equation*}
	\rho\mapsto\sum_{j=1}^{n}v_jf^\infty_j, \qquad \rho\mapsto\frac{1}{\rho}\sum_{j=1}^{n}v_jf^\infty_j,
\end{equation*}
cf. also Eq.~\eqref{eq:q.u}, define analytically the fundamental and speed diagrams, respectively. In particular, if for any given $\rho\in[0,\,\rhomax]$ system~\eqref{eq:G-L.continuous.space} admits a unique stable equilibrium then these mappings are actual functions; otherwise, they define \emph{multivalued diagrams}, which have also been studied in the literature, cf. e.g.,~\cite{gunther2003SIAP}.

\subsection{A prototypical case study}
\label{sect:case.study}
In order to substantiate more the issues set forth above, we now specialize Eq.~\eqref{eq:spat.homog} by exemplifying a structure of the interaction rate and of the table of games. We point out that our choices will not be, generally speaking, the most refined possible ones from the modeling point of view. Here we rather aim at examining a sufficiently handy prototypical model, yet meaningful for the application to vehicular traffic, which allows us to state precisely some results interesting for the development of the discrete kinetic theory of traffic flow. As a matter of fact, the interaction rate and the table of games that we will propose are particular instances of those introduced in~\cite{fermo2013SIAP}. Interested readers are also referred to~\cite{bellouquid2012M3AS,delitala2007M3AS} for further examples.

As a legacy of the classical collisional gas-kinetic theory, the interaction rate $\eta_{hk}$ may depend on the speeds of the interacting particles (we recall that e.g., in the discrete Boltzmann equation one has $\eta_{hk}\propto\vert v_k-v_h\vert$ but see also~\cite{coscia2007IJNM}). However, considering that cars do not interact through collisions and that more or less frequent interactions are essentially due to the level of traffic congestion on the road, we assume that the frequency of car interactions is actually independent of the pre-interaction speeds and is instead proportional to the car density. Hence we set:
\begin{equation}
	\eta_{hk}\equiv \eta=\eta_0\rho,
	\label{eq:int.rate}
\end{equation}
where $\eta_0>0$ is a proportionality constant which can be hidden in the time scale of Eq.~\eqref{eq:spat.homog}.

For constructing the table of games $A_{hk}^{j}$ we have instead to assess, in probability, which events can induce speed transitions toward new speed classes. To this purpose, we distinguish two cases.
\begin{enumerate}
\item[(i)] If $v_h\leq v_k$, i.e., if the candidate vehicle is slower than or at most as fast as the field vehicle, we assume that the effect of the interaction is that the candidate vehicle either keeps its pre-interaction speed, with a probability increasing with the congestion of the road, or is motivated to accelerate to the next speed class, with a probability increasing with the free room available on the road. Thus we set:
\begin{subequations}
\begin{equation}
\fbox{$v_h\leq v_k$}\
\begin{cases}
	h\ne n & 
	\begin{cases}
		A_{hk}^h=\frac{\rho}{\rhomax} \\
		A_{hk}^{h+1}=1-\frac{\rho}{\rhomax} \\
		A_{hk}^j=0 \quad \text{if\ } j\ne h,\,h+1
	\end{cases} \\
	h=n &
	\begin{cases}
		A_{nn}^n=1 \\
		A_{nn}^j=0 \quad \text{if\ } j\ne n.
	\end{cases}	\
\end{cases}
\label{eq:tab.games.h<k}
\end{equation}
Notice that if the candidate vehicle is already in the highest speed class $v_n$ then it can only keep the pre-interaction speed with unit probability, due to the lack of further higher speed classes.
\item[(ii)] If instead $v_h>v_k$, i.e., if the candidate vehicle is faster than the field vehicle, we assume that after the interaction the candidate vehicle either keeps its speed, with a probability increasing with the free room available on the road (which corresponds, for instance, to the case in which the candidate vehicle can overtake the field vehicle) or decelerates to the speed class of the field vehicle, with a probability increasing with the congestion of the road (which, conversely, corresponds to the case in which the candidate vehicle cannot overtake and is forced to queue). Hence we set:
\begin{equation}
\fbox{$v_h>v_k$}\
\begin{cases}
	A_{hk}^k=\frac{\rho}{\rhomax} \\
	A_{hk}^h=1-\frac{\rho}{\rhomax} \\
	A_{hk}^j=0 \quad \text{if\ } j\ne h,\,k.
\end{cases}
\label{eq:tab.games.h>k}
\end{equation}
\end{subequations}
\end{enumerate}

\begin{figure}[t]
\begin{center}
\includegraphics[width=0.9\textwidth]{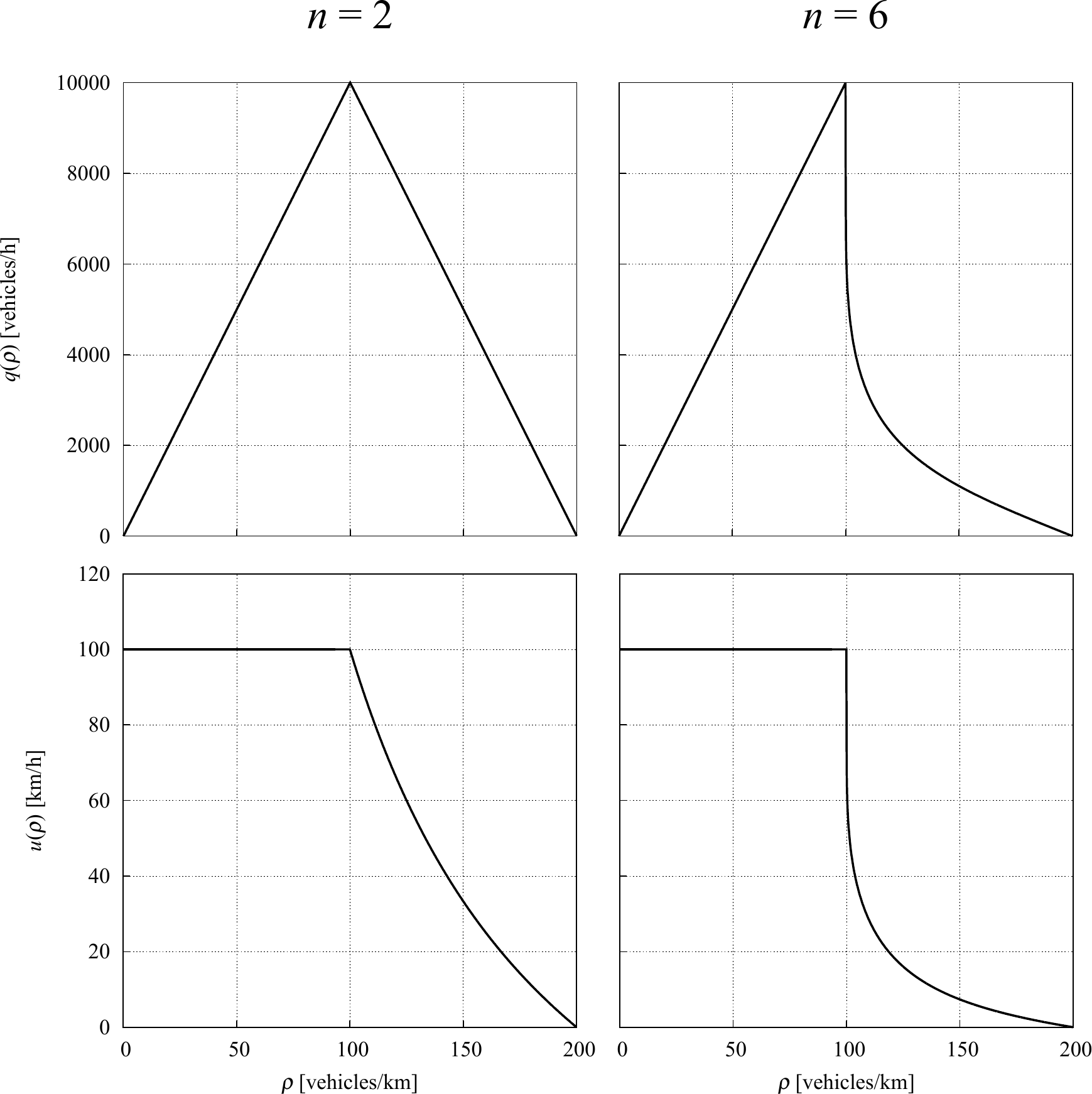}
\end{center}
\caption{Fundamental diagrams (top row) and speed diagrams (bottom row) obtained from model~\eqref{eq:spat.homog}--\eqref{eq:speed.lattice} using $n=2$ (left column) and $n=6$ (right column) speed classes.}
\label{fig:funddiag}
\end{figure}

Typical orders of magnitude of the maximum car density and speed in highway-type roads are $\rhomax=200\ \unit{vehicles/km}$ and $\vmax=100\ \unit{km/h}$. Using a uniformly spaced speed lattice in the interval $[0,\,\vmax]$ consisting of $n\geq 2$ speed classes:
\begin{equation}
	v_j=\frac{j-1}{n-1}\vmax, \quad j=1,\,\dots,\,n
	\label{eq:speed.lattice}
\end{equation}
we obtain from model~\eqref{eq:spat.homog}--\eqref{eq:speed.lattice} the fundamental and speed diagrams depicted in Fig.~\ref{fig:funddiag} for $n=2$ and $n=6$ speed classes.

We notice that the model catches the two main phases of traffic as theoretically described in the literature: a first \emph{free phase} at low density, in which the mean speed is constant at $\vmax$ and the flux is linear, followed by a \emph{congested phase} at high density, in which the mean speed and flux decrease monotonically, and possibly nonlinearly, to zero. In particular, with two speed classes only, namely $v_1=0$ and $v_2=\vmax$, the resulting fundamental diagram is the triangular one introduced by Daganzo~\cite{daganzo2005TRB} and Newell~\cite{newell1993TRB} as a simplification of more elaborated diagrams for both theoretical and numerical purposes, see also~\cite{blandin2009AMC}.  Moreover, it is interesting to observe that the speed diagram is in perfect agreement with that obtained numerically in~\cite{degond2008KRM} by simulating a \emph{microscopic} ``follow-the-leader'' model. This confirms practically that the kinetic approach successfully retains the microscopic character of car-to-car interactions, although the representation of the system is not focused on single vehicles. Increasing the number of speed classes produces a stronger nonlinear trend of both fundamental relationships in the congested phase, yet the critical density $\sigma$ at which the phase transition occurs is always $\sigma=\frac{\rhomax}{2}=100\ \unit{vehicles/km}$ (we will prove this property, among other ones, in the next Section~\ref{sect:theory}). This is ultimately due to the fact that our prototypical table of games~\eqref{eq:tab.games.h<k}-\eqref{eq:tab.games.h>k} lacks a parameter related to the environment (such as e.g., the quality of the road, weather conditions), which can account for different driver behaviors in different roads. For more precise models detailing also this aspect we refer interested readers to~\cite{bellouquid2012M3AS,fermo2013SIAP}. However, it is worth stressing that even such a simple model is able to predict the macroscopic phase transition as a consequence of more elementary lower scale car-to-car dynamics. Particularly, it does not require an \emph{ad hoc} construction in which such a transition is given as an external input to the model itself. To date, this seems to be impossible in a purely macroscopic approach, where the phase transition has to be postulated heuristically, cf. e.g.,~\cite{blandin2011SIAP,colombo2002SIAP}.

\section{Asymptotic analysis and equilibria}
\label{sect:theory}
In this section we study the qualitative properties of system~\eqref{eq:spat.homog}, particularly its asymptotic trends and equilibria which, as stated in Section~\ref{sect:space.homog}, are at the basis of the characterization of fundamental relationships. For the sake of simplicity, we consider the dimensionless version of the model in which the $f_j$'s and $\rho$ are scaled with respect to $\rhomax$, so that the maximum dimensionless density is $1$, and likewise the $v_j$'s are scaled with respect to $\vmax$, so that the maximum dimensionless speed is $v_n=1$. In practice, we use Eqs.~\eqref{eq:spat.homog}--\eqref{eq:speed.lattice} with $\rhomax=\vmax=1$.

To begin with, we notice that, owing to the assumption that the interaction rate is independent of the speeds of the interacting vehicles, cf. Eq.~\eqref{eq:int.rate}, the spatially homogeneous equations~\eqref{eq:spat.homog} can be rewritten as:
\begin{equation}
	\frac{df_j}{dt}=\eta[\rho]\left(\sum_{h,\,k=1}^{n}A_{hk}^{j}[\rho]f_hf_k-\rho f_j\right),
	\qquad j=1,\,\dots,\,n,
	\label{eq:simpl.spat.homog}
\end{equation}
where we have further stressed that both $\eta$ and the $A_{hk}^{j}$'s depend on $\rho$. By prescribing an initial condition $\f^0=(f^0_1,\,\dots,\,f^0_n)\in\R^n$ we obtain a Cauchy problem, which well-posedness has been established in~\cite{delitala2007M3AS}. We report here the statement of the result for completeness.
\begin{theorem}
Let $\f^0$ be such that 
\begin{equation*}
	f^0_j\geq 0 \quad \forall\,j=1,\,\dots,\,n, \qquad \sum_{j=1}^{n}f^0_j=\rho\in [0,\,1].
\end{equation*}
There exists a unique global solution $\f=(f_1,\,\dots,\,f_n)\in C^1((0,\,+\infty);\,\R^n)$ to Eq.~\eqref{eq:simpl.spat.homog} such that $\f(0)=\f^0$ and:
\begin{equation*}
	f_j(t)\geq 0 \quad \forall\,j=1,\,\dots,\,n, \qquad \sum_{j=1}^{n}f_j(t)=\rho
\end{equation*}
for all $t>0$.

Moreover, given two initial data $\f^0$, $\g^0$ with the same density $\rho$ the following \emph{a priori} continuity estimate holds true:
\begin{equation*}
	\norm{1}{\g(t)-\f(t)}+\norm{1}{\g'(t)-\f'(t)}\leq (1+3\eta_0)e^{3\eta_0t}\norm{1}{\g^0-\f^0},
		\qquad \forall\,t>0,
\end{equation*}
where $\eta_0:=\sup_{\rho\in[0,\,1]}\eta[\rho]$ and $\norm{1}{\cdot}$ is the $1$-norm in $\R^n$. In particular, the above estimate is uniform in time in every interval $[0,\,T]$, $T<+\infty$.
\label{theo:wellpos}
\end{theorem}

Since the solution to Eq.~\eqref{eq:simpl.spat.homog} is global in time, it makes sense to investigate its asymptotic behavior for $t\to +\infty$. A first result is proved again in~\cite{delitala2007M3AS} and concerns the existence of equilibria:
\begin{theorem}
For every $\rho\in [0,\,1]$ there exists at least one equilibrium point $\f^\infty=(f^\infty_1,\,\dots,\,f^\infty_n)\in\R^n$ of Eq.~\eqref{eq:simpl.spat.homog} such that:
\begin{equation}
	f^\infty_j\geq 0 \quad \forall\,j=1,\,\dots,\,n, \qquad \sum_{j=1}^{n}f^\infty_j=\rho.
	\label{eq:good.equilibrium}
\end{equation}
\label{theo:existence.equilibria}
\end{theorem}

Theorems~\ref{theo:wellpos} and~\ref{theo:existence.equilibria} are quite general, in that they do not assume any specific expression of $\eta[\rho]$ and $A_{hk}^{j}[\rho]$. As a matter of fact, the only really important properties (besides non-negativity, obvious for modeling purposes) are that $\eta[\rho]$ be bounded and the $A_{hk}^{j}[\rho]$'s satisfy Eq.~\eqref{eq:table.games.sum.1} for every $\rho\in[0,\,1]$.

A detailed expression of the table of games is instead necessary in order to characterize the uniqueness and stability of equilibria. Some studies on this issue have been performed in~\cite{delitala2007M3AS}, however they have been limited to a restricted number of speed classes, specifically $n=2,\,3$, while a result for generic $n\geq 2$ is still missing. The next theorem, proposed here for the first time, contributes to filling such a gap by taking advantage of the table~\eqref{eq:tab.games.h<k}-\eqref{eq:tab.games.h>k}.
\begin{theorem}
Let the transition probabilities be as in Eqs.~\eqref{eq:tab.games.h<k},~\eqref{eq:tab.games.h>k}. For every $n\geq 2$ and every $\rho\in [0,\,1]$ there exists a unique equilibrium point $\f^\infty$ of Eq.~\eqref{eq:simpl.spat.homog} satisfying~\eqref{eq:good.equilibrium} and which is also stable and attractive.
\label{theo:uniqueness.equilibria}
\end{theorem}
\begin{proof}
Let $\f^\infty$ be any of the equilibria that Theorem~\ref{theo:existence.equilibria} speaks of. If $\rho=0$ then the unique possibility is trivially $f^\infty_j=0$ for all $j=1,\,\dots,\,n$; therefore we will henceforth focus on the case $0<\rho\leq 1$.

First, we observe that it is sufficient to prove the uniqueness of the first $n-1$ components of $\f^\infty$, for then the $n$-th one is uniquely determined by the conservation of mass:
\begin{equation*}
	f^\infty_n=\rho-\sum_{j=1}^{n-1}f^\infty_j.
\end{equation*}

Equilibria of system~\eqref{eq:simpl.spat.homog} solve the equation:
\begin{equation}
	\sum_{h,\,k=1}^{n}A_{hk}^{j}[\rho]f^\infty_h f^\infty_k-\rho f^\infty_j=0,
	\label{eq:equation.equilibria.compact}
\end{equation}
that, owing to the argument above, we only need to consider for $j=1,\,\dots,\,n-1$. In order to follow the two cases $h\leq k$ and $h>k$ which define the table of games (cf. Eqs.~\eqref{eq:tab.games.h<k} and~\eqref{eq:tab.games.h>k}, respectively), we split Eq.~\eqref{eq:equation.equilibria.compact} as:
\begin{equation}
	\underbrace{\sum_{h=1}^{n}\sum_{k=h}^{n}A_{hk}^{j}[\rho]f^\infty_h f^\infty_k}_{h\leq k}+
		\underbrace{\sum_{h=2}^{n}\sum_{k=1}^{h-1}A_{hk}^{j}[\rho]f^\infty_h f^\infty_k}_{h>k}
			-\rho f^\infty_j=0
	\label{eq:equation.equilibria.split}
\end{equation}
and then we prove the uniqueness of a stable and attractive $\f^\infty$ by proceeding by induction on $j$.

\subsubsection*{Basis: $f_1^\infty$ is unique.}
If $j=1$ the only nonzero coefficients of the table of games are:
\begin{enumerate}[i)]
\item for $h\leq k$, $A_{1k}^1[\rho]=\rho$ ($k=1,\,\dots,\,n$);
\item for $h>k$, $A_{h1}^{1}[\rho]=\rho$ ($h=2,\,\dots,\,n$).
\end{enumerate}
Hence Eq.~\eqref{eq:equation.equilibria.split} specializes as:
\begin{equation*}
	\rho\left(\sum_{k=1}^{n}f^\infty_k+\sum_{h=2}^{n}f^\infty_h-1\right)f^\infty_1=0,
\end{equation*}
which, using $\sum_{k=1}^{n}f^\infty_k=\rho$ and consequently $\sum_{h=2}^{n}f^\infty_h=\rho-f^\infty_1$, gives
\begin{equation*}
	\rho(2\rho-1-f^\infty_1)f^\infty_1=0.
\end{equation*}
The two solutions are ${(f^\infty_1)}_{I}=0$ and ${(f^\infty_1)}_{II}=2\rho-1$. In order to study their stability, we notice that the calculations just made also imply that the right-hand side of Eq.~\eqref{eq:simpl.spat.homog} for $j=1$ is $\eta[\rho](-\rho f_1^2+\rho(2\rho-1)f_1)$, namely a second degree polynomial in $f_1$ with negative leading coefficient, whose roots are precisely the two possible values of $f^\infty_1$ found above. It is a basic fact of stability theory that, in such a case, the only stable equilibrium, which is also attractive, coincides with the larger root, whereas the other one is unstable. Hence the first component of $\f^\infty$ is univocally determined by choosing:
\begin{equation}
	f^\infty_1=
	\begin{cases}
		0 & \textup{if\ } \rho\leq\frac{1}{2} \\
		2\rho-1 & \textup{if\ } \rho>\frac{1}{2}.
	\end{cases}
	\label{eq:f1.equil}
\end{equation}

\subsubsection*{Inductive step: if $f^\infty_1,\,\dots,\,f^\infty_{j-1}$ are unique then so is $f^\infty_j$.}
For $2\leq j\leq n-1$ the only nonzero coefficients of the table of games are:
\begin{enumerate}[i)]
\item for $h\leq k$, $A_{jk}^{j}[\rho]=\rho$ ($k=j,\,\dots,\,n$) and $A_{j-1,k}^{j}[\rho]=1-\rho$ ($k=j-1,\,\dots,\,n$);
\item for $h>k$, $A_{hj}^{j}[\rho]=\rho$ ($h=j+1,\,\dots,\,n$) and $A_{jk}^{j}[\rho]=1-\rho$ ($k=1,\,\dots,\,j-1$),
\end{enumerate}
thus Eq.~\eqref{eq:equation.equilibria.split} specializes as:
\begin{multline*}
	\underbrace{\rho f^\infty_j\sum_{k=j}^{n}f_k+(1-\rho)f^\infty_{j-1}\sum_{k=j-1}^{n}f^\infty_k}_{h\leq k}+
		\underbrace{\rho f^\infty_j\sum_{h=j+1}^{n}f^\infty_h+(1-\rho)f^\infty_j\sum_{k=1}^{j-1}f^\infty_k}_{h>k} \\
			-\rho f^\infty_j=0.
\end{multline*}
Using the conservation of mass as before for replacing conveniently the sums up to $n$ with sums up to at most $j$ and rearranging the terms we finally discover:
\begin{equation}
	-\rho{(f^\infty_j)}^{2}
		+\left[(1-3\rho)\sum_{k=1}^{j-1}f^\infty_k+\rho(2\rho-1)\right]f^\infty_j
			+(1-\rho)f^\infty_{j-1}\left(\rho-\sum_{k=1}^{j-2}f^\infty_k\right)=0,
	\label{eq:equation.fj.equil}
\end{equation}
which, since by the inductive hypothesis $f^\infty_1,\,\dots,\,f^\infty_{j-1}$ are known, is a second degree polynomial equation in the unknown $f^\infty_j$. The discriminant
\begin{equation}
	\Delta={\left[(1-3\rho)\sum_{k=1}^{j-1}f^\infty_k+\rho(2\rho-1)\right]}^2
		+4\rho(1-\rho)f^\infty_{j-1}\left(\rho-\sum_{k=1}^{j-2}f^\infty_k\right)
	\label{eq:discriminant}
\end{equation}
being positive for all $0<\rho\leq 1$, there are two real roots whose only the larger one is a stable and attractive equilibrium, because the leading coefficient of the polynomial is negative. In addition, by inspecting the constant term after rewriting the polynomial in monic form (i.e., collecting $-\rho$ at the left-hand side) we infer that the product of the two roots is negative, hence the larger one is necessarily positive. Finally, we conclude that the admissible (viz. stable and nonnegative) $f^\infty_j$ is unique.
\end{proof}

From the proof of Theorem~\ref{theo:uniqueness.equilibria} we can deduce the following recursive formula for the components of the equilibrium $\f^\infty$:
\begin{equation*}
	f^\infty_j=
	\begin{cases}
		\begin{cases}
			0 & \textup{if\ } \rho\leq\frac{1}{2} \\
			2\rho-1 & \textup{if\ } \rho>\frac{1}{2}
		\end{cases}
	 		& \textup{if\ } j=1 \\[6mm] 
		\dfrac{-\left[(1-3\rho)\displaystyle{\sum_{k=1}^{j-1}}f^\infty_k+\rho(2\rho-1)\right]+\sqrt{\Delta}}{2\rho}
			& \textup{if\ } 2\leq j\leq n-1 \\ 
		\rho-\displaystyle{\sum_{k=1}^{n-1}}f^\infty_k
			& \textup{if\ } j=n,
	\end{cases}
\end{equation*}
where $\Delta$ is given by Eq.~\eqref{eq:discriminant}. Notice that for $2\leq j\leq n-1$ we have chosen the solution to Eq.~\eqref{eq:equation.fj.equil} with the positive sign in front of the square root of $\Delta$ because it is certainly the larger one.

If $\rho\leq\frac{1}{2}$ then it results $f^\infty_j=0$ for all $j=1,\,\dots,\,n-1$ and $f^\infty_n=\rho$, which implies that the flux and mean speed at equilibrium are, respectively, 
\begin{equation*}
	q(\rho)=v_nf^\infty_n=\rho, \qquad u(\rho)=\frac{1}{\rho}v_nf^\infty_n=1
\end{equation*}
regardless of the number $n$ of speed classes. This explains the same trend observed in Fig.~\ref{fig:funddiag} in the free phase for both $n=2$ and $n=6$. Notice that $u$ is actually indeterminate for $\rho=0$ but can be extended by continuity.

Conversely, if $\rho>\frac{1}{2}$ the trend of $q$ and $u$ depends on $n$, as it is also evident from Fig.~\ref{fig:funddiag} itself. In the case $n=2$ it is not difficult to compute explicitly these trends for all $\rho$. We have indeed:
\begin{equation*}
	f^\infty_1=
	\begin{cases}
		0 & \textup{if\ } \rho\leq\frac{1}{2} \\
		2\rho-1 & \textup{if\ } \rho>\frac{1}{2},
	\end{cases}
	\qquad
	f^\infty_2=
	\begin{cases}
		\rho & \textup{if\ } \rho\leq\frac{1}{2} \\
		1-\rho & \textup{if\ } \rho>\frac{1}{2}
	\end{cases}
\end{equation*}
and $v_1=0$, $v_2=1$, whence
\begin{equation*}
	q(\rho)=f^\infty_2=
	\begin{cases}
		\rho & \textup{if\ } \rho\leq\frac{1}{2} \\
		1-\rho & \textup{if\ } \rho>\frac{1}{2},
	\end{cases}
	\qquad
	u(\rho)=\frac{1}{\rho}f^\infty_2=
	\begin{cases}
		1 & \textup{if\ } \rho\leq\frac{1}{2} \\
		\frac{1}{\rho}-1 & \textup{if\ } \rho>\frac{1}{2},
	\end{cases}
\end{equation*}
which confirms the triangular-shaped flux with critical density $\sigma=\frac{1}{2}$ at road capacity (phase transition).

As anticipated in Section~\ref{sect:case.study}, the proof of Theorem~\ref{theo:uniqueness.equilibria} indicates that the critical density for model~\eqref{eq:spat.homog}--\eqref{eq:speed.lattice} is invariably $\sigma=\frac{1}{2}$, namely $\sigma=\frac{\rhomax}{2}$ in dimensional form, independently of the number of speed classes. In more detail, when $\rho$ crosses such a $\sigma$ a \emph{supercritical bifurcation} occurs: for $\rho<\sigma$ there actually exists only the stable equilibrium ${(f^\infty_1)}_{I}=0$, because ${(f^\infty_1)}_{II}=2\rho-1$ is negative and thus non admissible; for $\rho>\sigma$ the equilibrium ${(f^\infty_1)}_{I}$ still exists but becomes unstable, while the new stable equilibrium ${(f^\infty_1)}_{II}$ appears.

\bibliographystyle{plain}
\bibliography{FlTa-funddiag_traffic}

\end{document}